\documentclass[12pt]{article}

\usepackage[margin=1in]{geometry}
\usepackage[english]{babel}
\usepackage[utf8]{inputenc}

\usepackage{amsmath}
\usepackage{amssymb}
\usepackage{amsfonts}
\usepackage{amsthm}
\usepackage{mathdots}
\usepackage{mathrsfs}
\usepackage{bbm}

\usepackage{enumitem}
\usepackage{array}
\newcolumntype{C}{>{$}c<{$}}

\usepackage{graphicx}
\usepackage[dvipsnames,svgnames]{xcolor}
\usepackage[labelfont=bf,labelsep=period,justification=raggedright]{caption}
\usepackage{subcaption}
\usepackage[all]{xy}
\usepackage{tikz}
\usetikzlibrary{calc}
\usepackage{pgfplots}
\pgfplotsset{compat=1.18}
\usepackage{tkz-fct}

\usepackage{multicol}
\usepackage{textcomp}
\usepackage{url}
\usepackage{cite}
\usepackage{indentfirst}

\usepackage{hyperref}
\usepackage{aliascnt}
\usepackage[capitalize,nameinlink]{cleveref}

\crefname{enumi}{Part}{Parts}
\Crefname{enumi}{Part}{Parts}

\theoremstyle{plain}

\newtheorem{theorem}{Theorem}[section]
\crefname{theorem}{Theorem}{Theorems}
\Crefname{theorem}{Theorem}{Theorems}

\newaliascnt{thm}{theorem}

\aliascntresetthe{thm}
\crefname{thm}{Theorem}{Theorems}
\Crefname{thm}{Theorem}{Theorems}

\newaliascnt{lemma}{theorem}
\newtheorem{lemma}[lemma]{Lemma}
\aliascntresetthe{lemma}
\crefname{lemma}{Lemma}{Lemmas}
\Crefname{lemma}{Lemma}{Lemmas}

\newaliascnt{proposition}{theorem}

\aliascntresetthe{proposition}
\crefname{proposition}{Proposition}{Propositions}
\Crefname{proposition}{Proposition}{Propositions}

\newaliascnt{prop}{theorem}
\newtheorem{prop}[prop]{Proposition}
\aliascntresetthe{prop}
\crefname{prop}{Proposition}{Propositions}
\Crefname{prop}{Proposition}{Propositions}

\newaliascnt{corollary}{theorem}
\newtheorem{corollary}[corollary]{Corollary}
\aliascntresetthe{corollary}
\crefname{corollary}{Corollary}{Corollaries}
\Crefname{corollary}{Corollary}{Corollaries}

\newaliascnt{claim}{theorem}

\aliascntresetthe{claim}
\crefname{claim}{Claim}{Claims}
\Crefname{claim}{Claim}{Claims}

\newaliascnt{conjecture}{theorem}

\aliascntresetthe{conjecture}
\crefname{conjecture}{Conjecture}{Conjectures}
\Crefname{conjecture}{Conjecture}{Conjectures}

\theoremstyle{definition}

\newaliascnt{defn}{theorem}
\newtheorem{defn}[defn]{Definition}
\aliascntresetthe{defn}
\crefname{defn}{Definition}{Definitions}
\Crefname{defn}{Definition}{Definitions}

\newaliascnt{definition}{theorem}

\aliascntresetthe{definition}
\crefname{definition}{Definition}{Definitions}
\Crefname{definition}{Definition}{Definitions}

\theoremstyle{remark}

\newaliascnt{remark}{theorem}
\newtheorem{remark}[remark]{Remark}
\aliascntresetthe{remark}
\crefname{remark}{Remark}{Remarks}
\Crefname{remark}{Remark}{Remarks}

\newaliascnt{observation}{theorem}

\aliascntresetthe{observation}
\crefname{observation}{Observation}{Observations}
\Crefname{observation}{Observation}{Observations}

\newcommand{\eps}{\varepsilon}

\newcommand{\dtv}{d_\mathrm{TV}}
\newcommand{\Tmix}{T_\mathrm{mix}}
\newcommand{\wG}{\widehat{G}}
\newcommand{\wH}{\widehat{H}}

\newcommand{\EE}{\mathbb{E}}

\newcommand{\NN}{\mathbb{N}}

\newcommand{\mcD}{\mathcal{D}}

\newcommand{\mcS}{\mathcal{S}}

\newcommand{\bbone}{\mathbbm{1}}

\newcommand{\multiset}[2]{%
  \ensuremath{\left(\kern-.3em\left(\genfrac{}{}{0pt}{}{#1}{#2}\right)\kern-.3em\right)}%
}

\newcommand{\fptas}{\mathsf{FPTAS}}

\title{Sampling Simultaneous Edge-Colorings}

\author{
Ezra Furtado-Tiwari and 
Eric Vigoda\thanks{
Department of Computer Science, University of California, Santa Barbara.
Email: \{ezrafurtado-tiwari,vigoda\}@ucsb.edu.
Research supported in part by NSF grant CCF-2147094.
}
}

\date{}

\begin{document}
\maketitle

\begin{abstract}
We study the sampling problem for simultaneous edge colorings.  Given a pair of graphs $G_1=(V,E_1)$ and $G_2=(V,E_2)$ which are on the same vertex set $V$, a simultaneous edge coloring is an edge coloring of $G_1\cup G_2$ so that each of the individual graphs is properly colored.  When each of $G_1$ and $G_2$ are of maximum degree
$\Delta$, then it is conjectured that $\Delta+2$ colors suffice, and recent work asymptotically establishes the conjecture.  

We study Markov chains for randomly sampling from the uniform distribution over simultaneous edge colorings.  Straightforward applications of Jerrum's classical coupling argument establish rapid mixing of the Glauber dynamics on the corresponding line graph when $k>8\Delta$.  We present a simple weighted Hamming distance for which Jerrum's coupling yields optimal mixing time (up to constant factors) of $O(m\log{n})$ when $k>(6+\delta)\Delta$ for any fixed $\delta>0$.  Moreover, utilizing the flip dynamics with our new metric, we obtain $O(m\log{n})$ mixing of the flip dynamics when $k\geq 5.948\Delta$, using a local choice of flip parameters which only flips bounded-size components.  The proof adapts previous coupling analyses for the flip dynamics to the setting of simultaneous edge colorings.
\end{abstract}

\thispagestyle{empty}

\newpage

\setcounter{page}{1}

\section{Introduction}

A natural combinatorial problem which was recently introduced with intriguing open problems is simultaneous edge colorings. 
In the simultaneous edge coloring problem, we are given as input a pair of undirected graphs $G_1=(V,E_1)$ and $G_2=(V,E_2)$ on a common vertex set $V=[n]=\{1,\dots,n\}$, where each graph has maximum degree of $\Delta$.  A simultaneous edge-coloring is an assignment $\phi:E_1\cup E_2\rightarrow [k]$ such that for all $i\in\{1,2\}$, all $e,e'\in E_i$, if $e\cap e'\neq\emptyset$ then $\phi(e)\neq\phi(e')$.  In other words we are coloring the edges of both graphs so that the edges within each graph are properly colored, but adjacent edges in different graphs can be monochromatic.

Let $\chi(G_1,G_2)$ denote the minimum $k$ for which a simultaneous $k$-edge-coloring of $G_1$ and $G_2$ exists.
Vizing's Theorem shows that for a graph $G$ of maximum degree $\Delta$, its edge chromatic number satisfies $\chi(G)\leq\Delta+1$.  The intriguing question for simultaneous edge colorings is whether $\chi(G_1,G_2)\sim\Delta$.

The simultaneous edge coloring problem for graphs $G_1=(V,E_1)$ and $G_2=(V,E_2)$ can be recast as a vertex coloring problem via line graphs. Let $L_1$ and $L_2$ denote the line graphs of $G_1$ and $G_2$, respectively, and consider the graph $\wG = (\widehat{V}, \widehat{E})$ obtained by taking the union of $L_1$ and $L_2$, and identifying vertices corresponding to edges that appear in both $E_1$ and $E_2$.  A simultaneous edge coloring of $G_1$ and $G_2$ is then equivalent to a proper vertex coloring of~$\wG$.

If $G_1$ and $G_2$ have maximum degree $\Delta$, then each of $L_1$ and $L_2$ has maximum degree at most $2\Delta-2$, and hence $\wG$ has maximum degree at most $4\Delta-4$. However, this reduction obscures important structure: adjacency in $\wG$ arises only from constraints within $G_1$ or within $G_2$, and edges arising from different graphs impose no additional constraints on each other unless they correspond to the same original edge.  Thus, while the naive reduction suggests a dependence on $4\Delta$, the underlying constraint system retains a more refined structure tied to the original degree $\Delta$.

This observation raises the question of whether the mixing behavior of associated Markov chains for randomly sampling simultaneous edge colorings is governed by the maximum degree $4\Delta$ of the union graph, or by the intrinsic constraint degree $\Delta$ of the original graphs. Our results show that the latter is the case: optimal mixing time can be achieved in a regime determined by $\Delta$, rather than $4\Delta$.

The problem of simultaneous edge coloring was introduced recently by Cabello (see Bousquet and Durain~\cite{BousquetDurain}) who conjectured that $\chi(G_1,G_2)\leq\Delta+2$ for any graphs $G_1,G_2$ of maximum degree $\Delta$.
Recent work of Boyadzhiyska, Lang, Lo, and Molloy~\cite{BLLM}
established this threshold asymptotically by showing that $\chi(G_1,G_2)\leq \Delta + o(\Delta)$.

We address the corresponding sampling problem: for what range of $k$ vs.~$\Delta$ can we efficiently sample from the space of simultaneous $k$-edge-colorings.  Given graphs $G_1$ and $G_2$, and integer $k$, let $\Omega$ denote the set of simultaneous $k$-edge-colorings of $G_1$ and $G_2$.

The natural approach for the sampling problem is the Glauber dynamics for edge colorings. From an edge coloring $X_t \in \Omega$, we choose an edge $e \in E_1 \cup E_2$ uniformly at random, and a color $c \in [k]$ uniformly at random. If no edge adjacent to $e$ in $G_1$ or in $G_2$ has color $c$ in $X_t$, then we recolor $e$ to color $c$ to obtain $X_{t+1}$, and otherwise we set $X_{t+1} = X_t$.

The mixing time is the number of steps, from the worst initial state $X_0$, to reach within total variation distance $\leq 1/4$ of the stationary distribution.

The flip dynamics is a generalization of the Glauber dynamics for sampling vertex colorings of $\wG$, and hence simultaneous edge colorings of the pair $(G_1,G_2)$. The transitions ``flip'' maximal connected components induced by two colors, by interchanging the pair of colors appearing in the component.  The probability of flipping a component of size $\ell$ is $P_\ell/(mk)$ for specified flip probabilities $(P_i)_{i\geq 1}$, where $m$ is the number of vertices in $\wG$. When $P_1=1$ and $P_i=0$ for all $i\geq 2$ then  the flip dynamics coincides with the Glauber dynamics. In all previous works involving the flip dynamics~\cite{Vigoda00,CDMPP19,CarlsonVigoda25}, and also in the current work, the flip dynamics are 6-local which means that $P_i=0$ for all $i\geq 7$.  Consequently, $O(m\log n)$ mixing of a 6-local flip dynamics implies $O(m^2)$ mixing time for the Glauber dynamics via a comparison of the corresponding spectral gaps; moreover, the results of this paper, establishing a contractive coupling for list colorings, imply $O(m\log{n})$ mixing time of the Glauber dynamics for constant $\Delta$ via spectral independence~\cite{BCCPSV22,Liu21,CLV21,ALO}.

Sampling results for vertex colorings translate to edge colorings by taking the line graph and hence increasing the maximum degree $\Delta$ to $2\Delta$.  Jerrum~\cite{Jerrum} established $O(n\log{n})$ mixing time of the Glauber dynamics for vertex colorings whenever $k>2\Delta$, and hence  it establishes $O(m\log{n})$ mixing time for edge colorings (i.e., vertex colorings of the line graph) when $k>4\Delta$, and for simultaneous edge colorings when $k>8\Delta$, where $m$ is the number of edges in the original graph(s).

For general graphs, improved mixing results via the flip dynamics, which is a generalization of the Glauber dynamics, establish $O(n\log{n})$ mixing of the flip dynamics for vertex colorings when $k\geq 1.809\Delta$ by Carlson and Vigoda~\cite{CarlsonVigoda25} (see also, \cite{CDMPP19,Vigoda00}); this yields $O(m\log{n})$ mixing time when $k\geq 7.236\Delta$ for simultaneous edge colorings.  

The above fast mixing results of Jerrum~\cite{Jerrum} and the flip dynamics results of \cite{CarlsonVigoda25,CDMPP19,Vigoda00} show a contractive coupling, which then yields a deterministic approximate counting $\fptas$ algorithm, via the algorithm of Chen, Feng, Guo, Zhang, and Zou~\cite{CFGZZ25} for the same parameter range.

Further improved results~\cite{WZZ24-edge-colorings,CWZZ} are known for vertex colorings of a line graph, which corresponds to edge colorings of the original graph.  However these results do not directly apply to the problem of simultaneous edge colorings since it corresponds to vertex colorings of the union of two line graphs, as described earlier.

Our main contribution is a direct coupling analysis that exploits the overlapping constraint structure of simultaneous edge colorings. We introduce a weighted Hamming metric that distinguishes edges appearing in one graph versus both, enabling a refined path coupling argument which holds for all~$\Delta$.
 In contrast to recent approaches for edge colorings that rely on spectral independence and the matrix trickle-down framework, our analysis is based on an explicit coupling and avoids these more technical tools. As a consequence, our proofs are comparatively simple, and moreover yield a deterministic $\fptas$ for the approximate counting problem via the algorithm of Chen, Wang, Zhang, and Zhang~\cite{CWZZ}.
  In addition, our results hold for general graphs, including the unbounded degree case where $\Delta$ is allowed to grow with $n=|V|$.

We first present a simple and explicit coupling argument for the Glauber dynamics for simultaneous edge colorings when $k>6\Delta$; this improves upon the naive $k>8\Delta$ obtained via a generic reduction to Jerrum's $k>2\Delta$ result.  
We then extend our argument to the flip dynamics to obtain fast mixing when $k\geq 5.948\Delta$, improving upon the $k\geq 7.236\Delta$ obtained by applying~\cite{CarlsonVigoda25}.  As in previous works~\cite{Vigoda00,CDMPP19,CarlsonVigoda25}, our flip dynamics is 6-local which means that only components of size at most $6$ are flipped.

\begin{theorem}\label{thm:main-theorem}For all $n$, all $\Delta$, all pairs of $n$-vertex graphs $G_1=(V,E_1), G_2=(V,E_2)$ with maximum degree $\Delta$, the following hold, where $m$ is the number of edges in $G_1 \cup G_2$:
\begin{enumerate}[label=(\roman*)]
\item 
\label{main:easy}
For all fixed $\delta>0$, for $k>(6+\delta)\Delta$, the Glauber dynamics for simultaneous edge colorings has mixing time $O(m\log n)$.
\item \label{main:hard}
For $k \geq 5.948\Delta$,  the flip dynamics for simultaneous edge colorings has mixing time $O(m\log n)$ for a 6-local setting of the flip probabilities.
\end{enumerate}
\end{theorem}

In particular, both dynamics achieve $O(m\log n)$ mixing time in these parameter regimes.
Hayes and Sinclair~\cite{HayesSinclair} proved that for the Glauber dynamics on vertex colorings of the line graph, the mixing time is $\Omega(m\log{n})$ for any bounded degree graph.
Therefore, for bounded degree graphs, our bounds achieve optimal mixing time (up to constant factors) in regimes significantly below the $8\Delta$ threshold obtained via generic reductions.  

The proof of \cref{main:easy} also establishes $O(mk\log{n})$ mixing time when $k\geq 6\Delta+1$.  In addition, \cref{thm:main-theorem} holds for the list colorings version of simultaneous edge colorings, and consequently, as mentioned earlier, the algorithm of  \cite{CFGZZ25} applies and consequently we obtain a deterministic approximate counting algorithm for the partition function, see \cref{section:list-colorings} for further details.

We now briefly describe the main technical idea underlying our proofs.

Our proofs are based on a weighted Hamming distance in which edges are weighted according to whether they appear in one graph or in both graphs, distinguishing edges in $E_1 \cap E_2$ from those appearing in only one graph. This reweighting captures the asymmetric constraint structure of the problem and allows us to apply Jerrum's $k>2\Delta$ path coupling argument~\cite{Jerrum} to obtain a simple proof of contraction for the Glauber dynamics when $k>6\Delta$. We then adapt Vigoda's original $k>(11/6)\Delta$ analysis for the flip dynamics~\cite{Vigoda00} to this setting, yielding our improved bound for the flip dynamics.

We establish the setup and preliminaries for our analysis in \cref{section:prelims}.  We prove rapid mixing of the Glauber dynamics, thereby establishing \cref{main:easy} of \cref{thm:main-theorem}, in \cref{section:glauber-analysis}. In \cref{section:flip-theorem,section:flip-analysis,section:proofs}, we analyze the flip dynamics chain to prove \cref{main:hard} of \cref{thm:main-theorem}. Finally, we extend our results to list colorings in \cref{section:list-colorings}; this yields an $\fptas$ for the partition function.

\section{Preliminaries}\label{section:prelims}

\subsection{Mixing Time of a Markov Chain}

We will analyze the mixing times on our chain, which is a way to define the rate of convergence of our Markov chain to its stationary distribution. In our case this distribution will be the uniform distribution on the set of simultaneous edge colorings. 

Let $M$ be a Markov chain on state space $\Omega$ with transition matrix $P$ and stationary distribution $\pi.$
For probability distributions $\nu_1, \nu_2,$ on $\Omega$ we write the total variation distance as
$$\dtv(\nu_1,\nu_2) = \frac{1}{2}\sum_{\sigma \in \Omega} |\nu_1(\sigma) - \nu_2(\sigma)|.$$

\begin{defn}
    Define the \emph{mixing time} of $M$ for $\eps > 0$ to be 
    $$\Tmix(\eps) = \max_{X_0 \in \Omega} \min \{t \colon \dtv(P^t(X_0, \cdot), \pi) \leq \eps\}.$$
    We will refer to $\Tmix = \Tmix(1/4)$ as the mixing time; note that 
    $$\Tmix(\eps) \leq \Tmix \cdot \lceil \log_2(1/\eps)\rceil$$
    for any $\eps > 0.$
\end{defn}

\subsection{Setting: Vertex Colorings of Appropriate Line Graph}
\label{sub:setup}

\begin{defn}
Let $G_1 = (V,E_1)$ and $G_2 = (V,E_2)$ be graphs on a common vertex set. A simultaneous edge $k$-coloring is a function $\chi \colon E_1 \cup E_2 \to [k]$ such that for each $i \in \{1,2\}$, any two edges $e_1,e_2 \in E_i$ that share an endpoint satisfy $\chi(e_1) \neq \chi(e_2)$.
\end{defn}

We are interested in the set of simultaneous edge colorings which can be recast as vertex colorings in the following manner.

\begin{defn}
Let $G_1 = (V,E_1)$ and $G_2 = (V,E_2)$ be graphs on a common vertex set $V=[n]$.
    Let $L(G_1)$ and $L(G_2)$ denote the corresponding line graphs.
    Let $\wG = (\widehat{V}, \widehat{E})$ be the graph obtained by taking the union of $L(G_1)$ and $L(G_2)$, identifying vertices corresponding to edges that appear in both $E_1$ and $E_2$.
\end{defn}

Since $G_1$ and $G_2$ share a common vertex set, an edge $e$ may appear in both graphs, and hence the line graphs $L(G_1)$ and $L(G_2)$ may share vertices.

The set of simultaneous edge colorings of $G_1$ and $G_2$ is the same as the set of proper vertex colorings of $\wG$.

Suppose that $G_1$ and $G_2$ have maximum degree $\Delta$. Then each of $L(G_1)$ and $L(G_2)$ has maximum degree at most $2\Delta-2$, and hence $\wG$ has maximum degree at most $4\Delta-4$.
Our setting provides additional structure: if an edge appears in only one of $G_1$ and $G_2$, then the corresponding vertex in $\wG$ has degree at most $2\Delta$, while if it appears in both graphs its degree can be as large as $4\Delta-4$.

For the remainder of the argument we consider only vertex colorings of $\wG.$ Our state space $\Omega$ is thus the set of all proper vertex $k$-colorings of $\wG.$

\subsection{Path Coupling}

Our arguments for rapid mixing of both the Glauber and flip dynamics rely on the path coupling technique, which gives us estimates on the mixing time if we can provide a coupling which is a contraction mapping relative to some quasi-metric along a set of canonical paths between states. Suppose $\sigma, \xi \in [k]^{|\widehat{V}|}.$ We say that $\sigma$ and $\xi$ are \emph{neighbors} if they differ only at one vertex, and write $\sigma \sim \xi$ in this case. We say $\eta = (\eta_0, \ldots, \eta_\ell)$ is a \emph{simple path} if $\eta_i \neq \eta_j$ for $i \neq j$ and $\eta_i \sim \eta_{i+1}$ for all $i \in \{0,\ldots,\ell - 1\}.$ Denote the set of paths between $\sigma$ and $\xi$ as follows:
$$\rho(\sigma, \xi) = \{\eta \colon \eta_0 = \sigma, \eta_{\ell} = \xi, \eta \text{ is a simple path}\}.$$

We state a version of the path coupling lemma which is not in full generality but suffices for our argument.

\begin{lemma}[Bubley and Dyer, \cite{BubleyDyer97}]\label{lemma:path-coupling}

    Consider a Markov chain with state space~$\Omega^*.$ Let $\Sigma = \{(\sigma, \xi) \in [k]^{|\widehat{V}|} \colon \sigma \sim \xi\} \subset \Omega^* \times \Omega^*$, and let
    $$\Phi \colon \Omega^* \times \Omega^* \to \{0, 1, \ldots, D\}$$
    be such for all $(\sigma, \xi) \in \Omega^* \times \Omega^*,$  
    $$\Phi(\sigma, \xi) = \min\left\{ \sum_{i=0}^{\ell - 1}\Phi(\eta_i, \eta_{i+1}) \colon \eta \in \rho(\sigma, \xi)\right\}.$$
    Additionally we require that if $\sigma \neq \xi,$ $\Phi(\sigma, \xi) \geq 1.$

    If $(\Omega^*, \Sigma)$ is a connected graph, and there is a $\beta < 1$ and a coupling $(X_t, Y_t) \to (X_{t+1}, Y_{t+1})$ such that for $X_t \sim Y_t,$ 
    $$\EE[\Phi(X_{t+1}, Y_{t+1})] \leq \beta\Phi(X_t, Y_t),$$
    then we have the following upper bound on the mixing time:
    $$\Tmix(\eps) = O\left(\frac{\log(D/\eps)}{1 - \beta}\right).$$
\end{lemma}

Note that we only need to provide a coupling for the states $X_t \sim Y_t.$

\subsection{Expanded State Space}
\label{sub:expanded-state-space}

In order to apply the path coupling lemma to study chains on $\Omega,$ the set of all proper $k$-colorings of $\wG,$ we require that a simple path exists between any two states in $\Omega$. But note that this need not be true (for an easy counterexample take $C_4$ with $k = 2$ and try swapping colors $1$ and $2$). In order to ensure that we can apply the technique we expand our state space to $\widehat{\Omega} = [k]^{|\widehat{V}|}.$ We will extend the Glauber dynamics and flip dynamics chains to $\widehat{\Omega}$ such that if started at a proper coloring the chains only visit proper colorings, and if started at an improper coloring the chains eventually reach a proper coloring, since any monochromatic conflict can be resolved with positive probability. Thus any upper bound on the mixing time of the chains on the extended state space will be itself an upper bound on the mixing time $\Tmix$ of the natural chain on $\Omega.$ For this reason, from this point on we consider only the chains on $\widehat{\Omega}.$

\section{Warm-up: $k>6\Delta$ via Glauber}\label{section:glauber-analysis}

Our goal is to prove fast mixing of the Glauber dynamics for $k>6\Delta$, which is \cref{main:easy} of \cref{thm:main-theorem}.

As stated earlier, we will work via vertex colorings on $\wG.$ We first recall the definition of the Glauber dynamics chain for vertex $k$-colorings with expanded state space $\widehat{\Omega}$.

\begin{defn}[Glauber Dynamics]
    The Glauber Dynamics for vertex $k$-colorings on $\wG$ is the Markov Chain $(X_t)$ on the state space $\widehat{\Omega}$ where the transitions $X_t\rightarrow X_{t+1}$ are constructed from the following process.  For $X_t\in\widehat{\Omega}$:
    \begin{enumerate}
        \item Choose a vertex $v_t \in \widehat{V}$ and a color $c_t \in \{1, \ldots, k\}$ uniformly at random.
        \item For all $w \in \widehat{V},$ let 
        $$X_{t+1}(w) = \begin{cases} c_t & \text{ if } w = v_t \text{ and } c_t \not \in X_t(N(v_t)) \\ X_t(w) & \text{otherwise.}\end{cases}$$
    \end{enumerate}
\end{defn}

Note in particular that if $X_t \in \Omega$ then the condition in the second step is equivalent to requiring that $X_{t+1}\in\Omega$, i.e., that $X_{t+1}$ is a proper vertex $k$-coloring. The chain on $\Omega$ is ergodic for $k \geq 4\Delta - 2$ (recall, the Glauber dynamics is defined on the graph $\wG$ which has degree $\leq 4\Delta-4$), and the unique stationary distribution for the chain on $\Omega$ is the uniform distribution over $\Omega$, which are proper vertex $k$-colorings of $\wG$.
Furthermore, as alluded to in \cref{sub:expanded-state-space}, if $k \geq 4\Delta - 3$ then from an arbitrary $X_0 \in \widehat{\Omega}$ we will have $X_t \in \Omega$ for $t$ sufficiently large with probability $1$, hence the stationary distribution for the chain on $\widehat{\Omega}$ is also uniform over $\Omega$, and any upper bound on the mixing time of the Glauber dynamics on~$\widehat{\Omega}$ implies the same upper bound on the mixing time of the chain on~$\Omega$.

We proceed via a path coupling argument. Let $X_t,Y_t\in \widehat{\Omega}$ be a pair of (not necessarily proper) colorings where $X_t\oplus Y_t=\{v^{\ast}\}$ for some $v^{\ast}\in \widehat{V}$.  We provide the following quasimetric:

\begin{defn}
    Let the \emph{weight} of a vertex $w \in \widehat{V}$ be
    $$W(w) = \begin{cases} 2 & w \in E_1 \cap E_2 \\ 1 & \text{ otherwise.}\end{cases}.$$
    Note that we use $w$ to represent both a vertex in $\widehat{V}$ and the corresponding edge in $E_1 \cup E_2.$
    
    Let $\sigma$ and $\tau$ be vertex $k$-colorings of $\wG.$ Define the metric
    $$\wH(\sigma, \tau) = \sum_{v\in \widehat{V}: \sigma(v) \neq \tau(v)}W(v).$$
    Note that $\wH$ can be thought of as a weighted Hamming distance.
\end{defn}

\begin{prop}\label{prop:nbhd_weight}
Note that we have the bound
$$\sum_{w \in N(v)} W(w) \leq 4\Delta$$
for all vertices $v \in \wG$.
\end{prop}

The proposition follows from the fact that we can consider the contribution from $E_1$ and $E_2$ independently to the sum (since a vertex with weight $2$ corresponds to an edge in $E_1 \cap E_2$), and the degree of a vertex as stated previously is bounded above by $2\Delta$ in $L(G_1)$ and $L(G_2).$

We need to construct a coupling for which 
$$\EE[\wH(X_{t+1}, Y_{t+1})] \leq \beta \wH(X_t,Y_t)$$
for some $\beta < 1$ in order to apply the path coupling theorem. We use Jerrum's coupling method \cite{Jerrum}. In particular, the coupling will be constructed as follows:
\begin{enumerate}
    \item We will choose pairs $(v_t,c_t)$ and $(v_t,c_t')$ for some $c_t',$ and transition $X_t \to X_{t+1}$ and $Y_t \to Y_{t+1}$ by attempting to recolor $v_t$ with $c_t$ and $c_t'$, respectively.
    
    Choose the pair $(v_t,c_t) \in \widehat{V} \times [k]$ uniformly at random. We describe the process by which we choose $c_t'.$

    \item If $v_t = v^{\ast}$ then set $c_t' = c_t.$

    \item If $v_t \in N(v^{\ast})$, we have two cases. The first is if $c_t \not \in \{X_t(v^{\ast}), Y_t(v^{\ast})\}.$ In this case set $c_t' = c_t.$
    Otherwise $c_t = X_t(v^{\ast})$ or $c_t = Y_t(v^{\ast}).$ If $c_t = X_t(v^{\ast})$ then set $c_t' = Y_t(v^{\ast}).$ Otherwise set $c_t' = X_t(v^{\ast}).$
    \item Finally if $v \neq v^{\ast}$ and $v \not\in N(v^{\ast})$ then set $c_t' = c_t.$
    \item Attempt to recolor $v_t$ by $c_t$ in $X_{t}$ and $v_t$ by $c_t'$ in $Y_{t}.$ 
\end{enumerate}

Note that each color $c_t'$ is attempted for exactly one vertex, so the coupling is a bijection and thus $Y_t \to Y_{t+1}$ is a faithful copy of the Glauber dynamics chain (hence this is a valid coupling).

The idea is to greedily pair the the moves which are not equal in both chains, so as to minimize the probability that we increase $\wH.$ To prove the theorem it remains to analyze the change in $\wH$.

Our argument will analyze the increase in $\wH$ per vertex $z \in \widehat{V}.$ 
We introduce some notation to describe this change.

\begin{defn}\label{defn:dist_increase_glauber}
    Let $z \in \widehat{V}$ be a vertex. We define the random variable
    \[
\alpha_z =
[\wH(X_{t+1},Y_{t+1})-\wH(X_t,Y_t)]
\cdot \bbone(v_t=z).
\]
\end{defn}
Consequently, 
    $$\EE[\wH(X_{t+1}, Y_{t+1}) - \wH(X_t,Y_t)] = \sum_{z \in \widehat{V}}\EE[\alpha_z].$$

We can prove fast mixing for the Glauber dynamics.

\begin{proof}[Proof of \cref{thm:main-theorem}, \cref{main:easy}]
We will show that
\[
\EE[\wH(X_{t+1},Y_{t+1})-\wH(X_t,Y_t)]
\le
-\frac{k-6\Delta}{mk}\,\wH(X_t,Y_t).
\]

    Note that if $v_t \not \in \{v^{\ast}\} \cup N(v^{\ast})$ then our coupling does not change $\wH$ so we only need to consider $v_t = v^{\ast}$ and $v_t \in N(v^{\ast})$, i.e.
    \begin{equation}\label{glauber_proof:1}
    \sum_{z \in \widehat{V}}\EE[\alpha_z] = \sum_{z \in \{v^{\ast}\} \cup N(v^{\ast})}\EE[\alpha_z].
    \end{equation}
    
    First we analyze the case $v_t = v^{\ast}.$ Note that by construction of our coupling the attempted updates to $v_t$ succeed or fail in both chains. A successful recoloring occurs with probability at least 
    $$\frac{k - |\{X_t(w) \colon w \in N(v^{\ast})\}|}{mk} \geq \frac{k - |N(v^{\ast})|}{mk}$$
    and yields 
    $$\wH(X_{t+1}, Y_{t+1}) = 0 = \wH(X_t, Y_t) - W(v^{\ast}).$$

    We thus conclude
    \begin{equation}\label{glauber_proof:2}
        \EE[\alpha_{v^{\ast}}] \leq -W(v^{\ast})\left(\frac{k - |N(v^{\ast})|}{mk}\right).  
    \end{equation}

    Now we look at $v_t = w \in N(v^{\ast}).$ Note that $\wH$ only increases when $c_t = Y_t(v^{\ast}),$
    since otherwise $c_t = c_t'$ and both re-colorings either succeed or fail, meaning that $\wH$ either stays constant or decreases. The case $c_t = Y_t(v^{\ast})$ occurs with probability $\frac{1}{mk}$, and increases $\wH$ by at most $W(w).$ We thus have
    \begin{equation}\label{glauber_proof:3}
        \EE[\alpha_{w}] \leq \frac{W(w)}{mk}.
    \end{equation}

    It follows that 
    \begin{align*} 
\lefteqn{    \EE[\wH(X_{t+1}, Y_{t+1}) - \wH(X_t,Y_t)]
}
\hspace{.5in}
\\
&= \sum_{z \in \{v^{\ast}\} \cup N(v^{\ast})}\EE[\alpha_z] &\text{by }\labelcref{glauber_proof:1}\\
    &\leq \frac{1}{mk}\left(-W(v^{\ast})(k - |N(v^{\ast})|) + \sum_{w \in N(v^{\ast})} W(w)\right) &\text{by }\labelcref{glauber_proof:2} \text{ and } \labelcref{glauber_proof:3}\\
    &\leq \frac{1}{mk}\left(-W(v^{\ast})(k - |N(v^{\ast})|) + 4\Delta\right). &\text{by }\cref{prop:nbhd_weight}.\\
    \end{align*}

    We now have two cases depending on $W(v^{\ast}).$ 
  First if \(W(v^{\ast})=1\), then \(|N(v^{\ast})|\le 2\Delta\), so
\[
\EE[\wH(X_{t+1},Y_{t+1})-\wH(X_t,Y_t)]
\le
\frac{1}{mk}(-k+6\Delta)
=
-\frac{k-6\Delta}{mk}.
\]

In the second case \(W(v^{\ast})=2\), we have \(|N(v^{\ast})|\le 4\Delta\), so
\[
\EE[\wH(X_{t+1},Y_{t+1})-\wH(X_t,Y_t)]
\le
\frac{2}{mk}(-k+6\Delta)
=
-\frac{2(k-6\Delta)}{mk}.
\]
Thus, in both cases,
\[
\EE[\wH(X_{t+1},Y_{t+1})-\wH(X_t,Y_t)]
\le
-\frac{k-6\Delta}{mk}\,\wH(X_t,Y_t).
\]
Equivalently,
\[
\EE[\wH(X_{t+1},Y_{t+1})]
\le
\left(1-\frac{k-6\Delta}{mk}\right)\wH(X_t,Y_t).
\]
Since \(k>(6+\delta)\Delta\), we have
\[
\frac{k-6\Delta}{k}
=
1-\frac{6\Delta}{k}
>
\frac{\delta}{6+\delta},
\]
hence
\[
\EE[\wH(X_{t+1},Y_{t+1})]
\le
\left(1-\frac{\delta}{(6+\delta)m}\right)\wH(X_t,Y_t).
\]

    Applying the Path Coupling Lemma (\cref{lemma:path-coupling}) where $\beta = 1-\frac{\delta}{m(6+\delta)}$, we obtain
    $$\Tmix(\eps) = O(m \log(n/\eps)),$$
    which completes the proof of the theorem for the Glauber dynamics.
\end{proof}

\begin{remark}
    When $k\geq 6\Delta+1$ then the above proof approach establishes  $\Tmix(\eps) = O(mk \log(m/\eps)),$ since $\beta = 1-\frac{1}{mk}$ in this case.
\end{remark}

\section{Main Result via the Flip Dynamics}\label{section:flip-theorem}

We aim to produce a Markov chain on the set of simultaneous edge colorings of $(G_1,G_2)$ which mixes rapidly  when $k$ is a bit below $6\Delta$. To do this, we consider the flip dynamics on~$\wG$. We define the chain on $\widehat{\Omega}.$

\begin{defn}\label{defn:clusters}
    Let $\sigma \in \widehat{\Omega}$ be a (not necessarily proper) vertex $k$-coloring of $\wG$, and let $v \in \widehat{V}, c \in [k].$ Define the cluster
    $$S_{\sigma}(v, c) = \{w \in \widehat{V} \colon w \text { is reachable from } v \text{ by a } (\sigma(v), c)\text{-alternating path}\}.$$
    We take $S_{\sigma}(v, \sigma(v)) = \{v\}.$ Note that if a color $c$ does not appear in $N(v),$ then $S_{\sigma}(v,c) = \{v\}$ as well.

    Let $\mcS_{\sigma}$ be the set of all clusters in $\sigma.$
\end{defn}

We now define the flip dynamics chain.

\begin{defn}[Flip Dynamics]
    Let $\{P_i\}_{i \in \NN}$ be a set of probabilities such that $P_0 = 0, P_1 = 1$ and $0 \leq P_i \leq 1$ for all $i \ge 1.$ We say that a Markov Chain $(X_t)$ on the state space $\widehat{\Omega}$ is a setting of the Flip Dynamics chain for vertex $k$-colorings on $\wG$ if the transitions $X_t \to X_{t+1}$ are constructed via the following process:
    \begin{enumerate}
        \item Choose a pair $(v_t,c_t) \in \widehat{V} \times [k]$ uniformly at random.
        \item Let $s = |S_{X_t}(v_t, c_t)|.$ With probability $P_s/s$ interchange colors $X_t(v_t)$ and $c_t$ on the cluster $S_{X_t}(v_t,c_t)$, and let $X_{t+1}$ be the resulting coloring. Otherwise let $X_{t+1}= X_t.$
    \end{enumerate}

    We say that $(X_t)$ is $\ell$-local if $P_i = 0$ for all $i > \ell.$
\end{defn}

This is a generalization of the Glauber dynamics, which can be considered as the special case $P_2 = P_3 = \ldots = 0$ (since the Glauber dynamics only flips clusters of size $1$). Furthermore, note that by definition any transition in the Glauber dynamics chain occurs with nonzero probability in an arbitrary flip dynamics chain, so the flip dynamics chain is ergodic for $k \geq 4\Delta-2$ (as in the case of the Glauber dynamics).

Finally note that if $(X_t)$ is a flip dynamics chain and $X_t \in \Omega,$ then $X_{t+1} \in \Omega$ (and furthermore for $t$ sufficiently large, $\Pr[X_t \in \Omega] = 1$ for any $X_0$), hence again (as explained in \cref{sub:expanded-state-space}) an upper bound on the mixing time of the chain on $\widehat{\Omega}$ implies the same upper bound on the mixing time of the chain on $\Omega$.

In the setting of \cref{main:hard} of \cref{thm:main-theorem}, we use the flip probabilities 
\begin{equation} 
P_1 = 1,P_2 = \frac{137}{650},P_3 = \frac{77}{650},P_4 = \frac{47}{650},P_5 = \frac{27}{650},P_6 = \frac{12}{650}, P_i = 0 \mbox{ for }i \geq 7.
\label{flip:setting}
\end{equation}

As noted before, previous works utilizing the flip dynamics \cite{Vigoda00,CDMPP19,CarlsonVigoda25} also use 6-local flip probabilities (i.e., $P_i=0$ for all $i> 6$).  As pointed out in \cite{CDMPP19}, Vigoda's analysis \cite{Vigoda00} for $k>(11/6)\Delta$ requires $P_3=1/6$.   In \cite{CarlsonVigoda25} they set $P_3<1/6$, in particular, they have $P_3=.154$.  Our analysis yields $P_3$ substantially smaller than $1/6$; this results from the different metric but with similar extremal cases as in previous works, see \cref{rmk:worst_cases} for further discussion. 

\section{Analysis of the Flip Dynamics}\label{section:flip-analysis}
\subsection{Coupling Definition}
Let $X_t, Y_t \in \widehat{\Omega}$ be states in our flip dynamics chain such that $X_t \sim Y_t.$ Let $\{v^{\ast}\} = X_t \oplus Y_t$ be the vertex at which the colorings $X_t$ and $Y_t$ disagree. 

We now describe the coupling for moves in our Markov chain. For simplicity, we will consider the equivalent chain which at a state $\sigma \in \widehat{\Omega}$ chooses a cluster in $\mcS_{\sigma}$ according to the distribution in which $S_{\sigma}(v,c)$ (with $|S_{\sigma}(v,c)| = s$) is chosen with probability $\frac{s}{mk},$ and flips it with probability $P_{s}/s$ (note that this version of the chain is not efficiently implementable, but is equivalent for the purpose of analyzing the mixing time). Using this model, to provide a coupling it suffices to produce a fractional matching between the sets $\mcS_{X_t}$ and $\mcS_{Y_t}.$ 

We use the coupling of \cite{Vigoda00}, which greedily maximizes the overlap between matched clusters. Our coupling will be the identity coupling on $\mcS_{X_t} \cap \mcS_{Y_t},$ so it will be useful to describe the set $\mcS_{X_t} \oplus \mcS_{Y_t}$ of clusters which differ in $X_t$ and $Y_t.$ Since $X_t$ and $Y_t$ differ only at~$v^{\ast},$ these are precisely the components whose structure depends on the color at $v^*$, i.e.
\begin{multline*}
\mcD = \{S_{X_t}(v^{\ast}, X_t(w)) \colon w \in N(v^{\ast})\} \cup \{S_{Y_t}(v^{\ast}, X_t(w)) \colon w \in N(v^{\ast})\} \\ \cup \{S_{X_t}(w, Y_t(v^{\ast})) \colon w \in N(v^{\ast})\} \cup \{S_{Y_t}(w, X_t(v^{\ast})) \colon w \in N(v^{\ast})\}.
\end{multline*}

In simpler language, these are the clusters which involve a neighbor of $v^{\ast}$ and either $X_t(v^{\ast})$ or $Y_t(v^{\ast})$ as the other color. As a result, our coupling and analysis will be done by considering each color in the neighborhood of $v^{\ast}$, since for each color in the neighborhood there will be a cluster in either $X_t$ or $Y_t$ containing all neighbors with that color. 

Our goal will be to bound the increase in $\wH$ per color in the neighborhood of $v^{\ast},$ which will help us to establish that our coupling is a contraction mapping with respect to this metric. However, our introduction of a weighted structure on $\wG$ means that simply minimizing the increase in $\wH$ per color (as in the classical vertex colorings argument) is insufficient, as the degree constraints and total neighborhood weight of a vertex $v \in \widehat{V}$ depend on $W(v^{\ast})$ itself. 
Instead, we prove per-color bounds on the increase in $\wH$ relative to the total weight $W_c$ of the neighbors of $v^{\ast}$ with color $c,$ and then combine them with \cref{prop:nbhd_weight}. 

We are now ready to formally describe the coupling. We aim to describe the process which defines the transition $X_t \to X_{t+1}, Y_t \to Y_{t+1}.$ We will do this by describing for each color which components in $\mcD$ are matched with each other.

\begin{defn}
For a color $c \in [k],$ let $d_c = |\{w \in N(v^{\ast}) \colon X_t(w) = c\}|$ be the number of neighbors of $v^{\ast}$ with color $c.$ 
\end{defn}
We start by describing the coupling for the case $d_c = 1,$ and then describe the coupling for $d_c \geq 2.$ Fix $c \in [k]$ such that $d_c = 1.$ Note that in both the definitions the probabilities we give must be normalized by multiplication by $\frac{1}{mk}$; we provide the probabilities before normalization to simplify the definition. 

In this case we have a unique $w \in N(v^{\ast})$ with $X_t(w) = Y_t(w) = c$ and the clusters in $\mcD$ involving $c$ are the following.  In chain $X_t$ we have the pair of clusters $S_{X_t}(v^{\ast}, c)$ and $S_{X_t}(w,Y_t(v^{\ast}))$; and in chain $Y_t$ we have the pair of clusters $S_{Y_t}(v^{\ast},c)$ and $S_{Y_t}(w,X_t(v^{\ast}))$. 

\begin{enumerate}
    \item Let $A = |S_{X_t}(v^{\ast},c)|$ and $B = |S_{Y_t}(v^{\ast},c)|$.  
    \\ Note, $|S_{X_t}(w,Y_t(v^{\ast}))|=B-1$ and $|S_{Y_t}(w,X_t(v^{\ast}))|=A-1$. 
    \item \label{flip-both-X}
    Flip $S_{X_t}(v^{\ast},c)$ in $X_t$ and $S_{Y_t}(w,X_t(v^{\ast}))$ in $Y_t$ with probability $P_A$.
    \item \label{flip-both-Y}
    Flip $S_{X_t}(w,Y_t(v^{\ast}))$ in $X_t$ and $S_{Y_t}(v^{\ast},c)$ in $Y_t$ with probability $P_B$.
    \item Flip $S_{X_t}(w,Y_t(v^{\ast}))$ in $X_t$ and $S_{Y_t}(w,X_t(v^{\ast}))$ in $Y_t$ with probability $\min(P_{A-1} - P_A, P_{B-1} - P_B).$
    \item \label{flip-only-X}
    Flip $S_{X_t}(w,Y_t(v^{\ast}))$ in $X_t$ with probability $(P_{B-1} - P_B) - \min(P_{A-1} - P_A, P_{B-1} - P_B).$
    \item \label{flip-only-Y} 
    Flip $S_{Y_t}(w, X_t(v^{\ast}))$ in $Y_t$ with probability $(P_{A-1} - P_A) - \min(P_{A-1} - P_A, P_{B-1} - P_B).$
\end{enumerate}

Note that the component $S_{Y_t}(w,X_t(v^{\ast}))$ is $S_{X_t}(v^{\ast},c) \setminus \{v^{\ast}\},$ so what we are doing is flipping the large component in $X_t$ through $w$ and $v^{\ast},$ and then flipping the same component (with the exception of $v^{\ast}$) in $Y_t,$ so that on the overlap the components remain the same. The flip of the large component in $Y_t$ is analogously matched with the smaller component in $X_t.$ Now the smaller components are flipped with higher probability than the larger ones, so there will some remaining probability for both the smaller components. We greedily match these two components together, saving the increase at $w.$ Finally, with remaining probability we match the other components arbitrarily.

The case $d_c \geq 2$ is analogous, except we have clusters which involve $d_c$ neighbors in $X_t$ and $Y_t$; we choose to match the cluster in $X_t$ with the largest of the $d_c$ components corresponding to the neighbors in $Y_t$, and the cluster in $Y_t$ with the largest of the clusters corresponding to the neighbors in $X_t.$ Again we greedily match the clusters corresponding to the same neighbor, and then arbitrarily match the remaining ones.

    \begin{enumerate}
        \item \label{coupling:1} Let $A = |S_{X_t}(v^{\ast},c)|$ and $B = |S_{Y_t}(v^{\ast},c)|.$ Index the neighbors $w \in N(v^{\ast}) \cap \{u \in \widehat{V} \colon X_t(u) = c\}$ by $w_1, w_2, \ldots, w_{d_c}.$ Let $a_i = |S_{Y_t}(w_i, X_t(v^{\ast}))|$ and $b_i = |S_{X_t}(w_i, Y_t(v^{\ast}))|.$ If multiple neighbors $w_{i_1}, \ldots, w_{i_k}$ are such that the clusters $S_{Y_t}(w_{i_j}, X_t(v^{\ast}))$ coincide, then set $a_{i_j}= 0$ for $j \neq 1.$ Similarly, if $w_{i_1}, \ldots, w_{i_k}$ are such that the clusters $S_{X_t}(w_{i_j}, Y_t(v^{\ast}))$ coincide, then set $b_{i_j}= 0$ for $j \neq 1.$  With this definition, we have
        $$A = 1 + \sum_{i \colon a_i \neq 0}a_i,$$
        and
        $$B = 1 + \sum_{i \colon b_i \neq 0}b_i.$$ Furthermore, let $m_a = \arg\max_i a_i$ and $m_b = \arg\max_i b_i.$
        \item \label{coupling:2} With probability $P_{A}$ flip the component $S_{X_t}(v^{\ast},c)$ in $X_t$ and $S_{Y_t}(w_{m_a}, X_t(v^{\ast}))$ in $Y_t$.
        \item \label{coupling:3} With probability $P_{B}$ flip the component $S_{Y_t}(v^{\ast},c)$ in $Y_t$ and $S_{X_t}(w_{m_b}, Y_t(v^{\ast}))$ in $X_t$.
        \item For all \(i=1,\ldots,d_c\):
        \begin{enumerate}
        \item Let $q_i = P_{a_i} - \bbone(i = m_a) \cdot  P_A$ and $q_i' = P_{b_i} - \bbone(i = m_b) \cdot P_B.$ 
        \item \label{coupling:4} Flip $S_{Y_t}(w_i, X_t(v^{\ast}))$ in $Y_t$ and $S_{X_t}(w_i, Y_t(v^{\ast}))$ in $X_t$ with probability $\min(q_i, q_i').$
        \item \label{coupling:5} Flip $S_{Y_t}(w_i, X_t(v^{\ast}))$ in $Y_t$ with probability $q_i - \min(q_i, q_i').$
        \item \label{coupling:6} Flip $S_{X_t}(w_i, Y_t(v^{\ast}))$ in $X_t$ with probability $q_i' - \min(q_i, q_i').$
        \end{enumerate}
    \end{enumerate}

\subsection{Coupling Analysis}
To prove our theorem, we aim to bound the increase in $\wH$ over all transitions $X_t \to X_{t+1}, Y_t \to Y_{t+1}$ with $X_t \sim Y_t.$ In the proof of rapid mixing of the Glauber dynamics we analyzed the change in $\wH$ per neighbor of $v^{\ast}.$ In the case of the flip dynamics, we need to analyze the coupling in a per color manner due to the coupling definition. We thus define an analogous notion to \cref{defn:dist_increase_glauber} used in the proof of \cref{thm:main-theorem}, \cref{main:easy} for the Glauber dynamics.

\begin{defn}
Let \(c\in[k]\). Define the random variable
\[
\alpha_c =
\bigl(\wH(X_{t+1},Y_{t+1})-\wH(X_t,Y_t)\bigr)
\cdot
\bbone\bigl((v_t,c_t)=(v,c)\text{ for some }v\in\widehat V\bigr).
\]
\end{defn}

With this definition,
\[
\EE[\wH(X_{t+1},Y_{t+1})-\wH(X_t,Y_t)]
=
\sum_{c=1}^k \EE[\alpha_c].
\]

To bound the increase in $\wH$ over all configurations we will work only with \cref{prop:nbhd_weight}. The analogous argument to the Glauber dynamics would be to work with $|N(v^{\ast})|,$ but this can be either up to $2\Delta$ or $4\Delta$ depending on $W(v^{\ast}),$ so the result does not follow directly from this approach.

If we bound the maximum increase in $\wH$ per unit weight, we get a bound overall for the increase in $\wH$ after one move of the chain, by multiplying this increase by the $4\Delta$ upper bound. We introduce some notation to make this notion precise.
\begin{defn}
    Fix $c \in [k].$ Let 
    $$W_c = \sum_{w \in N(v^{\ast}) : \atop X_t(w) = c}W(w).$$
\end{defn}

To complete the argument it remains to bound $\EE[\alpha_c]$ for each $c \in [k].$ 
We first describe a set of conditions under which our results hold. These properties in particular hold for our choice of the flip probabilities.

\begin{prop}~\label{prop:weight-properties}
    \begin{enumerate}
    \item \label{property:constdif} The inequality
    $$(i-1)(P_{i} - P_{i+1}) \leq (P_2 - P_3)$$
    holds for all $2 \leq i \leq 6.$
    \item \label{property:maxat1} We have
    $$(W(w) + 2\ell(i-1))(P_{i} - P_{i+1}) \leq W(w)(P_1 - P_2)$$
    for all $2 \leq i \leq 6,$ $W(w) \in \{1,2\},$ $\ell \in \{1,2\}.$
    \item \label{property:decreasing}
    $P_i \geq P_{i+1} + P_{i+2}$ for all $i.$
    \item \label{property:prod_decreasing} 
    $iP_i \geq (i+1)P_{i+1}$ for all $i.$
    \item\label{property:maxat3} $(i-1)P_i \leq 2P_3$ for all $i.$
    \item\label{property:maxat4} $(2i-5)P_i \leq 3P_4$ for all $i.$
    \item\label{property:ineq1} $P_1 + P_2 - 2P_3 \leq \frac{3}{4} + 2P_3.$
    \item\label{property:ineq2} $8P_3 \leq P_1 + P_2 - 2P_3.$
    \item\label{property:ineq3} $\frac{1}{4} + 6P_3 \leq \frac{3}{4} + 2P_3.$
\end{enumerate}
\end{prop}

 In the upcoming analysis, the cases \(c=X_t(v^*)\) and \(c=Y_t(v^*)\), which potentially arise in the case of improper colorings, are included in the
bounds below. In these cases one of the corresponding large components
may be the trivial component \(\{v^*\}\), but the same greedy matching
argument applies, as in \cite{Vigoda00}, and gives no larger
contribution than the bounds stated in the following lemmas.

The following lemmas bound the expected change in $\wH$ on a per color basis in terms of~$d_c.$ 

\begin{lemma}\label{lemma:dc0}
    If $d_c = 0$ then 
    $mk\,\EE[\alpha_c] \leq -W(v^{\ast}).
    $
\end{lemma}

\begin{lemma}\label{lemma:dc-positive}
If $d_c\ge 1$, then
\[
mk\,\EE[\alpha_c]
\le
A(v^{\ast})W_c+W(v^{\ast})(d_c-1),
\]
where
\[
A(v^{\ast})=
\begin{cases}
\frac34+2P_3 & \mbox{if } W(v^{\ast})=1 \\
P_1+P_2-2P_3 & \mbox{if } W(v^{\ast})=2.
\end{cases}
\]
\end{lemma}

We will prove the above lemmas in \cref{section:proofs}. We first present the proof of our main result.

\begin{proof}[Proof of \cref{main:hard} of \cref{thm:main-theorem}]
Denote the degree of vertex $v^{\ast}$ as
$$D := |N(v^{\ast})| =\sum_{c = 1}^{k}d_c.$$
Let 
$$R:=\sum_{c:d_c\geq 1}(d_c-1).$$
Then, the number of colors appearing in the neighborhood of $v^{\ast}$ is $|\{c:d_c>0\}|=D-R$, and hence the number of available colors for $v^{\ast}$ is \begin{equation}
    \label{eqn:avail}
    |\{c:d_c = 0\}|= k - (D - R).
\end{equation}

Considering first the colors $c$ where $d_c=0$, we have by \cref{lemma:dc0} that
\[
mk\sum_{c:d_c=0}\EE[\alpha_c]
\le
-W(v^*)(k-D+R).
\]
For $d_c\geq 1$, by \cref{lemma:dc-positive},
\[
mk\,\EE[\alpha_c]
\le
A(v^*)W_c+W(v^*)(d_c-1).
\]
Combining these bounds gives
\begin{align*}
mk\,\EE[\wH(X_{t+1},Y_{t+1})-\wH(X_t,Y_t)]
&=
mk\sum_{c=1}^k\EE[\alpha_c]\\
&\le
-W(v^*)(k-D+R)
+\sum_{c:d_c\ge1}
\left(A(v^*)W_c+W(v^*)(d_c-1)\right)\\
&=
-W(v^*)(k-D+R)
+
A(v^*)\sum_{c:d_c\ge1}W_c
+
W(v^*)R\\
&=
-W(v^*)(k-D)
+
A(v^*)\sum_{w\in N(v^*)}W(w)\\
&\le
-W(v^*)k+W(v^*)D+4A(v^*)\Delta.
\end{align*}

If $W(v^*)=1$, then \(D\le2\Delta\), hence
\begin{align}
\nonumber
mk\,\EE[\wH(X_{t+1},Y_{t+1})-\wH(X_t,Y_t)]
& \le
-k+2\Delta+4\left(\frac34+2P_3\right)\Delta
\\ & \leq
-k+\frac{1933}{325}\Delta,
    \label{bound:W1}
\end{align}
where the last inequality follows from 
\[
2+4\left(\frac34+2P_3\right)
=
2+4\left(\frac34+\frac{154}{650}\right)
=
2+4\cdot\frac{1283}{1300}
=
\frac{1933}{325}.
\]

If \(W(v^*)=2\), then \(D\le4\Delta\) and \(A(v^*)=P_1+P_2-2P_3\), hence
\begin{align}
\nonumber
mk\,\EE[\wH(X_{t+1},Y_{t+1})-\wH(X_t,Y_t)]
& \le
2\left(-k+\left(4+2(P_1+P_2-2P_3)\right)\Delta\right)
\\ & \leq
2\left(-k+\frac{1933}{325}\Delta\right),
    \label{bound:W2}
\end{align}
where the last inequality follows from 
\[
4+2(P_1+P_2-2P_3)
=
4+2\left(1+\frac{137}{650}-2\cdot\frac{77}{650}\right)
=
4+2\cdot\frac{633}{650}
=
\frac{1933}{325}.
\]

Combining \cref{bound:W1,bound:W2}, we conclude the following:
\[
\EE[\wH(X_{t+1},Y_{t+1})]
\le
\left(1-\frac{k-\frac{1933}{325}\Delta}{mk}\right)\wH(X_t,Y_t).
\]
Since $5.948-\frac{1933}{325}=\frac{1}{3250}$, for $k\ge 5.948\Delta$ we find
\[
\EE[\wH(X_{t+1},Y_{t+1})]
\le
\beta\wH(X_t,Y_t),
\] where $\beta =  1- \frac{1/3250}{5.948} \cdot \frac{1}{m}$. Finally, by the Path Coupling Lemma (\cref{lemma:path-coupling}) we have 
\[
\Tmix(\varepsilon)=O(m\log(n/\varepsilon)).
\]
\end{proof}

\begin{remark}\label{rmk:worst_cases}
Note in particular that the cases that achieve the maximum increase in $\wH$ do in fact respect the degree constraints imposed by the structure of $\wG$, so that the bound on the mixing time is a tight bound. These cases are the $[(3,1)]$ and $[(3,3),(1,1)]$ configurations with $W(v^{\ast}) = 2$ and $W(v^{\ast}) = 1,$ respectively, as previously stated. The flip probabilities are chosen so that the expected increase in $\wH$ from both of these cases is equal.
\end{remark}

\section{Proofs of Lemmas}\label{section:proofs}

In this section we prove the required lemmas for the application of the path coupling theorem in the proof of \cref{main:hard} of \cref{thm:main-theorem}.

We turn our attention back to the coupling of \cref{section:flip-analysis}. Note that the increase in $\wH$ from move \labelcref{coupling:2} is bounded above by 
$$\sum_{i \colon a_i \neq 0}\bbone (i \neq m_a) \cdot (W(w_i) + 2(a_i-1))$$
and from move \labelcref{coupling:3} it is bounded above by
$$\sum_{i \colon b_i \neq 0}\bbone (i \neq m_b) \cdot (W(w_i) + 2(b_i-1)),$$
since $W(v) \leq 2$ for all $v \in \widehat{V}.$

Write $f(w_\ell)$ for the increase in $\wH$ from moves \labelcref{coupling:4}, \labelcref{coupling:5}, and \labelcref{coupling:6}. We get
$$f(w_\ell) \leq \max(q_{\ell}, q_{\ell}') \cdot W(w_{\ell}) + 2q_{\ell}(a_{\ell} - 1) + 2q_{\ell}'(b_{\ell}-1).$$

Additionally, write
$$W(A) := \sum_{i \colon a_i \neq 0}\left(W(w_i) + 2(a_i-1)\right)$$
and
$$W(B) := \sum_{j \colon b_j \neq 0}\left(W(w_j) + 2(b_j-1)\right).$$

With this setting, we provide proofs of the lemmas required for the proof of \cref{main:hard} of \cref{thm:main-theorem}.

\begin{proof}[Proof of \cref{lemma:dc0}]
    We have $S_{X_t}(v^{\ast}, c) = S_{Y_t}(v^{\ast}, c) = \{v^{\ast}\},$ and our coupling is the identity coupling on these states (i.e. our coupling sets $X_{t+1}(v^{\ast}) = Y_{t+1}(v^{\ast}) = c$). Thus $\wH$ decreases by at least $W(v^{\ast})$, hence $\EE[\alpha_c] \leq -\frac{W(v^{\ast})}{mk}$ (we pick the pair $(v^{\ast}, c)$ with probability $\frac{1}{mk}$) and thus $mk\EE[\alpha_c] \leq -W(v^{\ast}).$
\end{proof}

To prove \cref{lemma:dc-positive}, we will prove the following statements which imply the lemma:
\begin{prop}\label{prop:dc-positive-subcases}~
\begin{enumerate}
    \item If $d_c = 1$ then $mk\EE[\alpha_c] \leq (P_1 + P_2 - 2P_3)W_c.$ \label{lemma:dc1}
    \item If $W(v^{\ast}) = 2$ and $d_c \geq 2$ then $mk\EE[\alpha_c] \leq (P_1 + P_2 - 2P_3)W_c + 2(d_c - 1).$\label{lemma:weight2}
    \item If $W(v^{\ast}) = 1$ and $d_c \geq 2$ then $mk\EE[\alpha_c] \leq \left(\frac{3}{4} + 2P_3\right)W_c + (d_c-1).$\label{lemma:weight1}
\end{enumerate}
\end{prop}
Note that \cref{lemma:dc1} implies the result in the case $W(v^{\ast}) = 1$ by \cref{property:ineq1}.

\begin{proof}[Proof of \cref{lemma:dc1}]
    Note that $A = a_1 + 1$ and $B = b_1 + 1.$ Without loss of generality suppose $P_{a_1} - P_{a_1+1} \geq P_{b_1} - P_{b_1+1}.$ Then 

    \begin{multline*}
    mk\EE[\alpha_c] = \max(P_{a_1}-P_{a_1+1}, P_{b_1}-P_{b_1+1})W(w_1) + 2(P_{a_1}-P_{a_1+1})(a_1-1)\\ + 2(P_{b_1}-P_{b_1+1})(b_1-1)\\ = (P_{a_1}-P_{a_1+1})(W(w_1) + 2(a_1-1)) + 2(P_{b_1}-P_{b_1+1})(b_1-1).\\
    \end{multline*}
    By \cref{property:constdif,property:maxat1,property:decreasing} of \cref{prop:weight-properties},
    $$mk\EE[\alpha_c] \leq \left(P_1 - P_2 + 2(P_2 - P_3)\right)W_c = \left(P_1 + P_2 - 2P_3\right)W_c.$$
\end{proof}

\begin{proof}[Proof of \cref{lemma:weight2}]
    Write 
    \begin{align*}
    g(w_{\ell}) &= \max(P_{a_\ell}, P_{b_{\ell}}) \cdot W(w_{\ell}) + 2(a_{\ell} - 1)P_{a_{\ell}} + 2(b_{\ell} - 1)P_{b_{\ell}}.\\
    \end{align*}

    If $m_{a} = m_{b} = \ell$ then 
    \begin{align*}
    f(w_{\ell}) &\leq \max(P_{a_\ell}, P_{b_{\ell}}) \cdot W(w_{\ell}) + 2(a_{\ell} - 1)(P_{a_{\ell}} - P_A) + 2(b_{\ell} - 1)(P_{b_{\ell}} - P_B)\\
    &= \max(P_{a_\ell}, P_{b_{\ell}}) \cdot W(w_{\ell}) + 2(a_{\ell} - 1)P_{a_{\ell}} + 2(b_{\ell} - 1)P_{b_{\ell}} - 2P_A(a_{m_a} - 1) - 2P_B(b_{m_b} - 1)\\
    &= g(w_{\ell}) - 2P_A(a_{m_a} - 1) - 2P_B(b_{m_b} - 1).\\
    \end{align*}

    If $m_a \neq m_b$ we have by a very similar argument that 
    \begin{align*}
        f(w_{m_a}) + f(w_{m_b}) \leq g(w_{m_a}) + g(w_{m_b}) - 2P_A(a_{m_a} - 1) - 2P_B(b_{m_b} - 1).\\
    \end{align*}
    Thus we conclude
    $$\sum_{\ell} f(w_{\ell}) \leq \sum_{\ell} g(w_{\ell}) - 2P_A(a_{m_a} - 1) - 2P_{B}(b_{m_b} - 1),$$
    so that 
    \begin{multline*}
    mk\EE[\alpha_c] \leq P_A(W(A) - W(w_{m_a}) - 4(a_{m_a} - 1))\\ + P_B(W(B) - W(w_{m_b}) - 4(b_{m_b} - 1)) + \sum_{\ell} g(w_{\ell}).
    \end{multline*}
    We find that $g(w_{\ell})$ is maximized at $a_{\ell} = 1, b_{\ell} = 3,$ and is equal to $4P_3$ for this setting. Furthermore this maximizes the term $P_A(W(A) - W(w_{m_a}) - 4(a_{m_a} - 1)).$ By checking the finitely many cases $a_\ell, b_\ell \leq 6$ we find that the marginal decrease in $\sum_{\ell}g(w_{\ell})$ required to make the term $P_B(W(B) - W(w_{m_b}) - 4(b_{m_b} - 1))$ nonzero is greater than the contribution of this term to the sum, hence
    \begin{equation}\label{eq:d_cgeq2bound}
    mk\EE[\alpha_c] - W(v^{\ast})(d_c-1) \leq 2P_A(W(A) - W(w_{i_{\max}})) + \sum_{\ell}W(w_{\ell}) + 4d_cP_3 - W(v^{\ast})(d_c-1).
    \end{equation}
    We compute
    \begin{align*}
        mk\EE[\alpha_c] - 2(d_c - 1)&\leq \left(\frac{2P_A(W(A) - W(w_{m_a})) + 4d_cP_3 - 2(d_c-1) + W_c}{W_c}\right)W_c\\
         &\leq \left(\frac{2P_A(W(A) - W(w_{m_a})) + 4d_cP_3 - 2(d_c-1)}{W_c} + 1\right)W_c\\
         &\leq \left(\frac{2P_A(W(A) - W(w_{m_a})) + 4d_cP_3 + 2}{d_c} - 1\right)W_c\\
         &\leq \left(\frac{2P_A(W(A) - W(w_{m_a}))}{d_c} + \frac{4d_cP_3 + 2}{d_c} - 1\right)W_c\\
         &\leq \left(\frac{2P_A(W(A) - W(w_{m_a}))+2}{d_c} + 4P_3-1\right)W_c\\
         &\leq \left(\frac{2P_A(2(A - 1))+2}{d_c} + 4P_3-1\right)W_c\\
         &\leq \left(\frac{8P_3+2}{d_c} + 4P_3-1\right)W_c\\
         &\leq 8P_3W_c\\
         &\leq (P_1 + P_2 - 2P_3)W_c & \mbox{by \cref{property:ineq2}}.\\
    \end{align*}
\end{proof}

\begin{proof}[Proof of \cref{lemma:weight1}]
    First suppose that $d_c = 2.$ Without loss of generality suppose that $(P_{a_1} - P_A) \leq (P_{b_1} - P_B)$ and $m_a = 1.$ Note that since $A \geq a_1 + 2$ and $B \geq b_1 + 2$ the statement $(P_{a_1} - P_A) \leq (P_{b_1} - P_B)$ implies that $a_1 \geq b_1$ due to \cref{property:decreasing} of \cref{prop:weight-properties}.

    We compute 
    $$f(w_1) = \begin{cases}(P_{b_1} - P_B)W(w_ 1) + 2(a_1 - 1)(P_{a_1} - P_A) + 2(b_1 - 1)(P_{b_1} - P_B) & \text{ if } m_b = 1\\
    P_{b_1}W(w_1) + 2(a_1 - 1)(P_{a_1} - P_A) + 2(b_1 - 1)P_{b_1} & \text{ otherwise}\end{cases}$$
    and
    $$f(w_2) = \begin{cases}\max(P_{a_2}, P_{b_2} - P_B)W(w_2) + 2(a_2 - 1)P_{a_2} + 2(b_2 - 1)(P_{b_2} - P_B) & \text{ if } m_b = 2\\
    \max(P_{a_2}, P_{b_2})W(w_2) + 2(a_2 - 1)P_{a_2} + 2(b_2 - 1)P_{b_2} & \text{ otherwise}\end{cases}.$$
    We now have two cases depending on the weights $W(w_1)$ and $W(w_2)$.
    The first is when $W(w_1) \geq W(w_2).$ Note that $b_1$ and $b_2$ contribute the same expression to $f(w_1)$ and $f(w_2)$ respectively, except for the maximum term, which is always maximized when $b_1 \geq b_2,$ so it suffices to assume that $b_1 \geq b_2.$ Then our equations become
    $$f(w_1) = (P_{b_1} - P_B)W(w_1) + 2(a_1 - 1)(P_{a_1} - P_A) + 2(b_1 - 1)(P_{b_1} - P_B)$$
    and
    $$f(w_2) = \max(P_{a_2}, P_{b_2})W(w_2) + 2(a_2 - 1)P_{a_2} + 2(b_2 - 1)P_{b_2}.$$

    We thus get
    \begin{multline*}
    mk\EE[\alpha_c] = (2(a_2 - a_1) + W(w_2))P_A + (2(b_2 - b_1) + W(w_2) - W(w_1))P_B \\ + 2(a_1 - 1)P_{a_1} + 2(b_1 - 1)P_{b_1} + 2(a_2 - 1)P_{a_2} + 2(b_2 - 1)P_{b_2} + P_{b_1}W(w_1) + \max(P_{a_2}, P_{b_2})W(w_2).\\ 
    \end{multline*}

    This expression is maximized at $a_1 = a_2$ and $b_1 = b_2,$ since we know $a_1 \geq a_2$ and $b_1 \geq b_2.$ It follows that 
    \begin{multline*}
    mk\EE[\alpha_c]-1 \leq W(w_2)P_{2a_1+1} + (W(w_2) - W(w_1))P_{2b_1 + 1}\\ + 4(a_1 - 1)P_{a_1} + 4(b_1 - 1)P_{b_1} + (W(w_2)+W(w_1))P_{b_1}-1.
    \end{multline*}
    By \cref{property:maxat1} we have $b_1 = 1,$ and the terms involving $a_1$ are maximized at $a_1 = 3.$ In this case we get
    $$mk\EE[\alpha_c]-1 \leq (W(w_1) + W(w_2))P_1+ (W(w_2) - W(w_1))P_3 + 8P_3 - 1.$$
    It follows that 
    $$mk\EE[\alpha_c]-1 = \left(\frac{mk\EE[\alpha_c]-1}{W_c}\right)W_c \leq \left(\frac{8P_3 + 3}{4}\right)W_c = \left(\frac{3}{4} + 2P_3\right)W_c.$$

    Now we handle the other case. We know that $W(w_1) = 1$ and $W(w_2) = 2.$ If $b_1 \geq b_2$ the proof and bound is identical to the first case, so we can assume that $b_2 \geq b_1.$

    We have
    $$f(w_1) = P_{b_1}W(w_1) + 2(a_1 - 1)(P_{a_1} - P_A) + 2(b_1 - 1)P_{b_1}$$
    and
    $$f(w_2) = (P_{b_2} - P_B)W(w_2) + 2(a_2 - 1)P_{a_2} + 2(b_2 - 1)(P_{b_2} - P_{B}),$$
    so that
    \begin{multline*}
    mk\EE[\alpha_c] - 1 = (2(a_2 - a_1) + W(w_2))P_A + (2(b_1 - b_2) + W(w_2) - W(w_1))P_B \\ + 2(a_1 - 1)P_{a_1} + 2(b_1 - 1)P_{b_1} + 2(a_2 - 1)P_{a_2} + 2(b_2 - 1)P_{b_2} + W(w_1)P_{b_1} + W(w_2)P_{b_2} - 1.\\ 
    \end{multline*}
    It follows by the same analysis as the other case that $a_1 = a_2,$ but we cannot conclude that $b_1 = b_2,$ since we have the additional terms $W(w_1)P_{b_1}$ and $W(w_2)P_{b_2}.$ But then
    \begin{align*}
    mk\EE[\alpha_c]-1 &\leq W(w_2)P_{2a_1 + 1} + 4(a_1 - 1)P_{a_1} + 2(b_1 - 1)P_{b_1} + 2(b_2 - 1)P_{b_2} 
    \\ & \ \ \ \ \ \ \ + W(w_1)P_{b_1} + W(w_2)P_{b_2} - 1\\
    &\leq 8P_{3} + 3 - 1   
    & \mbox{by \cref{property:maxat1,property:maxat3}} \\
    &= 8P_{3} + 2.\\
    \end{align*} 
    We thus find
    $$mk\EE[\alpha_c]-1 = \left(\frac{mk\EE[\alpha_c]-1}{W_c}\right)W_c \leq \left(\frac{8P_3+2}{3}\right)W_c < \left(\frac{3}{4} + 2P_3\right)W_c.$$
    This proves the result for $d_c = 2.$

    We now consider the case $d_c \geq 3.$ Note that the bound given by \cref{eq:d_cgeq2bound} applies in out setting as well (since the proof of this bound does not require any bound on $d_c$ or $W(v^{\ast})$. We thus get 
    \begin{align*}
         mk\EE[\alpha_c] - (d_c-1) &\leq \left(\frac{2P_A(W(A) - W(w_{m_a})) + 4d_cP_3 - (d_c-1) + W_c}{W_c}\right)W_c\\
         &\leq \left(\frac{2P_A(W(A) - W(w_{m_a})) + 1}{W_c} + \frac{4d_cP_3 - d_c}{d_c} + 1\right)W_c\\
         &\leq \left(\frac{2P_A(W(A) - W(w_{m_a})) + 1}{W_c} + 4P_3\right)W_c\\
    \end{align*}
    At this point we have two cases to consider. The first is $d_c = 3$ and $W(w_1) = W(w_2) = W(w_3) = 1,$ so that $W_c = 3.$ We get 
    \begin{align*}
         mk\EE[\alpha_c] - (d_c-1) &\leq \left(\frac{2P_A(W(A)-1) + 1}{3} + 4P_3\right)W_c\\
         &\leq \left(\frac{2P_A(4 + 2(A-4)-1) + 1}{3} + 4P_3\right)W_c\\
         &\leq \left(\frac{2P_A(2A - 5) + 1}{3} + 4P_3\right)W_c\\
         &\leq \left(2P_4 + \frac{1}{3} + 4P_3\right)W_c &\mbox{by \cref{property:maxat4}}\\
         &< \left(\frac{1}{4} + 6P_3\right)W_c\\
         &\leq \left(\frac{3}{4} + 2P_3\right)W_c &\mbox{by \cref{property:ineq3}}\\
    \end{align*}
    Otherwise $W_c \geq 4,$ so that
    \begin{align*}
         mk\EE[\alpha_c] - (d_c-1) &\leq \left(\frac{2P_A(2(A-1)) + 1}{W_c} + 4P_3\right)W_c\\
         &\leq \left(\frac{4P_A(A-1) + 1}{W_c} + 4P_3\right)W_c\\
         &\leq \left(\frac{8P_3 + 1}{4} + 4P_3\right)W_c &\mbox{by \cref{property:maxat3}}\\
         &\leq \left(\frac{1}{4} + 6P_3\right)W_c\\
         &\leq \left(\frac{3}{4} + 2P_3\right)W_c &\mbox{by \cref{property:ineq3}}\\
    \end{align*}
    This completes the proof.
\end{proof}

\section{List Colorings}\label{section:list-colorings}

\subsection{Main Results}
In this section we expand the analysis of the flip dynamics on vertex colorings on $\wG$ to list colorings. This allows us to prove analogous results in the list coloring setting:
\begin{theorem}\label{thm:main-theorem-list} For all $n$, all $\Delta$, all pairs of $n$-vertex graphs $G_1=(V,E_1), G_2=(V,E_2)$ with maximum degree $\Delta$ and $k \geq 5.948\Delta$,  the flip dynamics for simultaneous edge list colorings has mixing time $O(m\log n)$ for a 6-local setting of the flip probabilities.
\end{theorem}

\begin{corollary}
    For all pairs of $n$-vertex graphs $G_1 = (V,E_1), G_2 = (V, E_2)$ with maximum degree $\Delta$ and $k \geq 5.948\Delta,$ there exists an FPTAS for the number of simultaneous edge colorings on $(G_1,G_2).$  
\end{corollary}
\begin{proof}
    Since \cref{thm:main-theorem-list} yields a contractive coupling for list colorings of $\wG$, we can apply \cite[Lemma 19]{CFGZZ25} to obtain coupling independence, which allows us to apply \cite[Theorem 4]{CFGZZ25} to conclude the existence of an FPTAS for vertex $k$-colorings on $\wG$. As discussed in \cref{sub:setup} this is equivalent to giving an FPTAS for counting simultaneous edge colorings.
\end{proof}
\subsection{Flip Dynamics for List Colorings}

We will expand the flip dynamics chain for vertex colorings on $\wG$ to list colorings on $\wG.$ We first describe the setup.

\begin{defn}
    A \emph{list assignment} on $\wG$ is a mapping $L \colon \widehat{V} \to 2^{\NN}.$ An $L$-coloring is an assignment $\sigma \colon \widehat{V} \to \NN$ such that $\sigma(v) \in L(v)$ for all $v \in \widehat{V}.$ We say that an $L$-coloring is proper if for all adjacent vertices $v_1, v_2 \in \widehat{V},$ $\sigma(v_1) \neq \sigma(v_2).$
\end{defn}

As in \cref{sub:expanded-state-space} we will consider $\Omega_L$ to be the set of proper $L$-colorings of $\wG$ and $\widehat{\Omega}_L$ the set of not necessarily proper $L$-colorings. We will define our Markov chain on $\widehat{\Omega}_L$ and by the same analysis as for the vertex colorings chains we can derive a mixing time bound for the chain on proper list colorings.

We now describe the flip dynamics chain on $\widehat{\Omega}_L.$ We define the clusters as in \cref{defn:clusters}, except in this case we can only flip the subset of clusters $S_{\sigma}(v, c)$ for which all $w \in S_{\sigma}(v,c)$ satisfy $\{\sigma(v), c\} \subseteq L(w).$ This allows us to swap all adjacent colors on the cluster. If this condition is satisfied we say that a cluster is \emph{flippable}. We let $\mcS^L_{\sigma}$ be the set of all flippable clusters in $\sigma.$ 

\begin{defn}[Flip Dynamics for List Colorings]
    Let $\{P_i\}_{i \in \NN}$ be a set of probabilities such that $P_1 = 1$ and $0 \leq P_i \leq 1$ for all $i \geq 1.$ We say that a Markov chain $(X_t)$ on the state space $\widehat{\Omega}_L$ is a setting of the Flip Dynamics chain for vertex $k$-list colorings on $\wG$ if $|L(v)| \geq k$ for all $v \in \widehat{V}$ and the transitions $X_t \to X_{t+1}$ are constructed via the following process:
    \begin{enumerate}
        \item Choose a pair $(v_t, i_t) \in \widehat{V} \times [k]$ uniformly at random. Let $c_t$ be the $i_t$-th smallest element of $L(v_t).$ 
\item        Let \(s = |S_{X_t}(v_t,c_t)|\). With probability \(P_s/s\),
interchange colors \(X_t(v_t)\) and \(c_t\) on the cluster
\(S_{X_t}(v_t,c_t)\), and let \(X_{t+1}\) be the resulting coloring.
Otherwise let \(X_{t+1}=X_t\).
    \end{enumerate}
\end{defn}
Again note that any transition in the Glauber dynamics chain on list colorings (which we have not defined here but is analogous to the Glauber dynamics on vertex colorings) occurs with nonzero probability in the flip dynamics chain, so this chain is ergodic for $k \geq 4\Delta-2$. In the setting of \cref{thm:main-theorem-list} we use the same flip probabilities as in \cref{thm:main-theorem}.

\subsection{Analysis of the Flip Dynamics for List Colorings}

Our coupling for the list colorings case is identical, except that we may flip some components with probability $0$ rather than $P_s,$ where $s$ is the size of the component. As a result, we will not state the coupling in full detail here, but we will point out the changes in the proof of rapid mixing.  

Rapid mixing of the chain follows almost directly from the proof of \cref{main:hard} of \cref{thm:main-theorem}.

\begin{proof}[Proof of \cref{thm:main-theorem-list}]
We will show that the bounds given by \cref{lemma:dc0,lemma:dc-positive} both hold when adapted to the list coloring case. In particular, the proof of the bound analogous to \cref{lemma:dc0} is exactly the same as before. 

For the bound given in \cref{lemma:dc-positive} corresponding to the case when $d_c = 1$ (which is \cref{lemma:dc1} of \cref{prop:dc-positive-subcases}), we have two additional cases to consider. Recall,  there are 4 clusters to consider when $d_c=1$, namely: in chain $X_t$ we have $S_{X_t}(v^{\ast}, c)$ and $S_{X_t}(w,Y_t(v^{\ast}))$, whereas in chain $Y_t$ we have $S_{Y_t}(v^{\ast},c)$ and $S_{Y_t}(w,X_t(v^{\ast}))$. 

Suppose that $c\not\in L(v^{\ast})$.  Hence, $S_{X_t}(v^{\ast}, c)$ and $S_{Y_t}(v^{\ast}, c)$ are not flippable, but $v^{\ast}$ has an additional available color, i.e., an additional color $c'$ where $d_{c'}=0$ (this is because $X_t(w)=Y_t(w)=c$ and $c\not\in L(v^{\ast})$). This case is equivalent to the $d_c = 2$ case, where the other neighbor $w'$ has clusters of size $0$ (or equivalently size $\geq 7$ since the chain is $6$-local), so that the probability of flipping clusters involving $w'$ is $0.$

Suppose that $c\in L(v^{\ast})$ but for some $z\in S_{X_t}(w,Y_t(v^{\ast}))$ we have $\{c,Y_t(v^{\ast})\}\not\subset L(z)$ (the case where for some $z'\in S_{Y_t}(w,X_t(v^{\ast}))$ we have $\{c,X_t(v^{\ast})\}\not\subset L(z')$ is analogous).  Hence, the clusters $S_{X_t}(w,Y_t(v^{\ast}))$ and $S_{Y_t}(v^{\ast},c)$ are not flippable, and hence it corresponds to setting $a=0$ (or equivalently to setting $a-1\geq 7$ since the chain is $6$-local), and then the remaining analysis holds as before. 

Otherwise all the components are flippable, in which case we can directly apply the regular vertex coloring analysis.

For the bounds in \cref{lemma:weight2,lemma:weight1} of \cref{prop:dc-positive-subcases}, corresponding to the cases when $d_c~\geq~2$, note that not flipping the component of size $A$ or $B$ only decreases $mk\EE[\alpha_c].$ The other change that may happen is that some of the clusters $S_{X_t}(w_i, Y_t(v^{\ast}))$ or $S_{Y_t}(w_i, X_t(v^{\ast}))$ corresponding to the neighbors $w_i$ are not flippable. Since our flip dynamics chain is $6$-local, this is equivalent to setting the corresponding $a_i$ or $b_i$ to $0$ (or equivalently setting the corresponding $a_i$ or $b_i$ to $\geq 7$), and then the analysis is identical to the proofs of \cref{lemma:weight2,lemma:weight1}.
\end{proof}

\bibliographystyle{alpha}
\bibliography{biblio}

\end{document}